\documentclass[amsmath,amssymb,aps,prx,notitlepage,footinbib,superscriptaddress,reprint,longbibliography]{revtex4-1}
\usepackage{graphicx}
\usepackage{dcolumn}
\usepackage{bm}
\usepackage{braket}
\usepackage{amsthm}
\usepackage{epstopdf}
\usepackage{amssymb}
\usepackage{amsmath}
\usepackage{bbold}
\usepackage{hyperref}
\usepackage{comment}
\usepackage{xcolor}
\usepackage[normalem]{ulem}
\usepackage[utf8]{inputenc}

\newcommand{\cor}[1]{\textcolor{red}{#1}}

\theoremstyle{definition}

\newtheorem{proposition}{Proposition}

\begin{document}
\preprint{APS/123-QED}
\title{Quantum Limits of Superresolution in a Noisy environment}
\author{Changhun Oh}%
\email{changhun@uchicago.edu}
\affiliation{Pritzker School of Molecular Engineering, The University of Chicago, Chicago, Illinois 60637, USA}
\author{Sisi Zhou}
\affiliation{Department of Physics, Yale University, New Haven, Connecticut 06511, USA}
\affiliation{Pritzker School of Molecular Engineering, The University of Chicago, Chicago, Illinois 60637, USA}
\author{Yat Wong}
\affiliation{Pritzker School of Molecular Engineering, The University of Chicago, Chicago, Illinois 60637, USA}
\author{Liang Jiang}%
\email{liang.jiang@uchicago.edu}
\affiliation{Pritzker School of Molecular Engineering, The University of Chicago, Chicago, Illinois 60637, USA}

\date{\today}
\begin{abstract}
We analyze the ultimate quantum limit of resolving two identical sources in a noisy environment.
We prove that in the presence of noise causing false excitation, such as thermal noise, the quantum Fisher information of arbitrary quantum states for the separation of the objects, which quantifies the resolution, always converges to zero as the separation goes to zero.
Noisy cases contrast with a noiseless case where it has been shown to be nonzero for a small distance in various circumstances, revealing the superresolution.
In addition, we show that false excitation on an arbitrary measurement, such as dark counts, also makes the classical Fisher information of the measurement approach to zero as the separation goes to zero.
Finally, a practically relevant situation resolving two identical thermal sources, is quantitatively investigated by using the quantum and classical Fisher information of finite spatial mode multiplexing, showing that the amount of noise poses a limit on the resolution in a noisy system.
\end{abstract}

\maketitle

The Rayleigh criterion poses a limit of resolution of two incoherent objects in classical optics \cite{rayleigh1879xxxi, born2013principles}.
Recently, inspired by quantum optics and quantum metrology, superresolution overcoming the Rayleigh limit has been proposed by replacing a conventional direct imaging technique with structured measurement techniques in a weak source regime \cite{tsang2016quantum}.
Since the breakthrough, the superresolution technique has been generalized for two incoherent thermal sources \cite{nair2016far}, arbitrary quantum states of two objects \cite{lupo2016ultimate}, two-dimensional imaging \cite{ang2017quantum}, three-dimensional imaging \cite{yu2018quantum, napoli2019towards}, estimating spatial deformation \cite{sidhu2017quantum}, and an arbitrary number of sources \cite{tsang2017subdiffraction, zhou2019modern, tsang2019quantum, lupo2020quantum}, and it has been also studied from the perspective of channel discrimination \cite{lu2018quantum, pirandola2019fundamental}.
Also, many proof-of-principle experiments have demonstrated that elaborately constructed measurements enable surpassing the Rayleigh limit in practice \cite{paur2016achieving, tang2016fault, yang2016far, tham2017beating, parniak2018beating}.
The main idea of revealing superresolution is to show that the quantum Fisher information (QFI) of two objects' separation, the inverse of which limits the estimation error of the separation, is still nonzero when the separation converges to zero.
This behavior contrasts with a conventional direct imaging method whose classical Fisher information (CFI) vanishes as the separation drops to zero, making the estimation error of the separation diverge for a small separation.

More recently, the effects of noise on superresolution techniques started to be analyzed.
The CFI of two point sources' separation using a spatial mode demultiplexing scheme has been shown to vanish in the presence of dark counts \cite{len2020resolution} or measurement crosstalk \cite{gessner2020superresolution} when the separation is small.
However, these analyses are restricted to specific measurement schemes.
Meanwhile, the QFI of resolving two incoherent thermal point sources also vanishes for small separations in the presence of thermal background noise \cite{lupo2020subwavelength}.
In this case, the influence of detection noise has not been systematically analyzed.
Besides, previous studies are limited to uncorrelated classical states such as thermal states and weak point sources.
More general quantum states need to be analyzed for applications on microscopy where we can manipulate quantum states of light emitted from sources to improve resolution.

In this Letter, we consider a more general situation of resolving two identical sources in arbitrary quantum states, assuming a generic noise model inevitable in experiments, which we call excitation noise.
Excitation noise is a type of noise causing false excitation that cannot be distinguished from signal photons, which includes thermal background noise and dark counts.
We first prove that excitation noise makes QFI vanish for small separations,
which indicates that the resolution of two close sources is inherently vulnerable to noise in practical imaging processes.
We also provide a quantitative analysis of noise in resolving two identical incoherent thermal sources.
We then show that the CFI of arbitrary measurement affected by excitation noise on detectors vanishes for small separations.
Notably, our results reproduce previous studies about the impact of noises on particular imaging processes and states \cite{len2020resolution, lupo2020subwavelength, gessner2020superresolution}.
Finally, we show that in the presence of thermal noise, a finite spatial mode demultiplexing (fin-SPADE) measurement is nearly optimal when the signal-to-noise ratio (SNR) is large.

\emph{The model.---}
Consider two identical sources with a separation $s>0$ that emit light described by two orthogonal creation operators $\hat{c}_{1,2}^\dagger$.
The emitted light reaches the image plane being attenuated such that $\hat{c}_{1,2}^\dagger\to\sqrt{\eta}\hat{a}_{1,2}^\dagger-\sqrt{1-\eta}\hat{u}_{1,2}^\dagger$ with environmental modes $\hat{u}_{1,2}^\dagger$ and being distorted as
\begin{align}
\hat{a}^\dagger_1\equiv\int_{-\infty}^\infty dx \psi(x-s/2)\hat{a}_x^\dagger, ~~ \hat{a}^\dagger_2\equiv\int_{-\infty}^\infty dx \psi(x+s/2)\hat{a}_x^\dagger.
\end{align}
Here, $\psi(x)$ represents the point-spread function (PSF) of the imaging system, assumed to be real for simplicity.
Also, the mode operators for different positions satisfy the canonical commutation relation (CCR) $[\hat{a}_x,\hat{a}^\dagger_{x'}]=\delta(x-x')$.
In general, the two mode operators do not obey the CCR since the two PSFs $\psi(x\pm s/2)$ have a nonzero overlap, i.e., $[\hat{a}_1,\hat{a}_2^\dagger]\neq 0$.
Thus, we define symmetric and antisymmetric modes $\hat{a}_\pm$ to orthogonalize them \cite{tsang2016quantum, nair2016far, lupo2016ultimate, yu2018quantum} as
\begin{align}
\hat{a}_\pm\equiv \frac{\hat{a}_1\pm\hat{a}_2}{\sqrt{2(1\pm\delta)}},~~~ \delta(s)\equiv \int_{-\infty}^\infty dx\psi(x+s/2)\psi(x-s/2),
\end{align}
which satisfy the CCR, i.e., $[\hat{a}_+,\hat{a}_-^\dagger]=0$.
Now, the overall dynamics can be captured as
\begin{align}\label{c_op}
\hat{c}_\pm^\dagger\equiv \frac{\hat{c}_1^\dagger\pm\hat{c}_2^\dagger}{\sqrt{2}}\to\sqrt{\eta_\pm}\hat{a}^\dagger_\pm-\sqrt{1-\eta_\pm}\hat{u}_\pm^\dagger,
\end{align}
where $\eta_\pm\equiv(1\pm\delta)\eta$ represent effective attenuation rates, and $\hat{u}_\pm$ represent auxiliary modes.
Furthermore, the imaging process of estimating the separation $s$ can be described by the following dynamics of the mode operators \cite{lupo2016ultimate} (See Appendix \ref{appendix:hamiltonian}),
\begin{align}\label{dynamics}
\frac{d\hat{a}_\pm}{ds}=i[\hat{H}_\pm^\text{eff},\hat{a}_\pm],
\end{align}
where the effective Hamiltonians are written as
\begin{align}\label{ham}
\hat{H}_\pm^\text{eff}=i\frac{d\theta_\pm}{ds}(\hat{c}_\pm^\dagger \hat{v}_\pm-\hat{c}_\pm\hat{v}_\pm^\dagger)-iB_\pm(\hat{a}_\pm\hat{b}^\dagger_\pm-\hat{a}_\pm^\dagger \hat{b}_\pm),
\end{align}
where $\hat{v}_\pm$ are the environmental mode operators before the transformation, $\theta_\pm\equiv \arccos\sqrt{\eta_\pm}$,
\begin{align}
\hat{b}_\pm&\equiv \frac{1}{B_\pm}\frac{\partial\hat{a}_\pm}{\partial s},~~~\text{and}~~~ B_\pm\equiv -\frac{\epsilon_\pm}{2\sqrt{1\pm\delta}}.
\end{align}
Thus, mode operators $\hat{b}_\pm$ represent the derivative of the spatial modes, $\hat{a}_\pm(s+ds)\approx \hat{a}_\pm(s)+\partial_s \hat{a}_\pm(s)ds$.
We have also defined the following parameters:
\begin{align}
\epsilon_\pm^2&\equiv \Delta k^2\mp\beta-\frac{\gamma^2}{1\pm\delta}, ~~\gamma\equiv \delta'(s), ~~ \Delta k^2\equiv \beta(0), \\
\beta(s)&\equiv-\delta''(s)=\int_{-\infty}^\infty dx\frac{d\psi(x+s/2)}{dx}\frac{d\psi(x-s/2)}{dx}.
\end{align}
Here, $\gamma$ represents the variation of the overlap from the changes of the separation $s$, $\Delta k^2$ accounts for the variance of the momentum operator $-i\partial_x$, and $\beta$ represents interference between the derivatives of the PSFs.
The effective Hamiltonians show that when the separation $s$ changes, the attenuation to the environment $\hat{u}_\pm$ varies and the derivative modes $\hat{b}_\pm$ are excited through the beam-splitter-like Hamiltonian, which is the last term in Eq.~\eqref{ham}.
Note that the model assumes that the light evolves under a passive transformation before reaching the image plane and that since we use the Heisenberg picture, the emitted light from sources can be an arbitrary quantum state.

\emph{QFI in a noisy system.---}
From the perspective of quantum metrology, resolution can be quantified by the QFI of the separation $s$ \cite{tsang2016quantum}.
QFI $H(\theta)$ of a quantum state $\hat{\rho}(\theta)$ for an unknown parameter $\theta$ gives a lower bound of the estimation error for $\theta$, $\Delta^2\theta\geq 1/MH(\theta)$, which is the so-called quantum Cram\'{e}r-Rao inequality \cite{helstrom1976quantum, holevo2011probabilistic, braunstein1994statistical, paris2009quantum}. Here, $M$ is the number of independent trials.
Note that the quantum Cram\'{e}r-Rao inequality indicates that the estimation error diverges if the QFI vanishes.

Before we present our main result, we define excitation noise as a type of noise that transforms any quantum state to be a  full rank state.
The physical interpretation of the noise is that it introduces false excitation indistinguishable from the signal.
Thermal background noise is such noise, which is described by a beam-splitter interaction with an environmental mode in a thermal state of a nonzero photon number \cite{walls2007quantum},
because thermal background noise transforms a state into a full rank state.

Now, we present our main result:
\begin{proposition}
For imaging processes in the presence of excitation noise, the QFI for the separation $s$ of two identical sources in arbitrary quantum states converges to zero as $s\to0$.
\end{proposition}
\begin{proof}
Let $\hat{\rho}(s)$ be an arbitrary quantum state of light at the image plane, emitted by two identical sources separated by $s$.
First, because the two objects are identical, replacing $s$ by $-s$ does not change the description of the system.
Thus, we have $d\hat{\rho}/ds\propto s \hat{\sigma}$ with a Hermitian operator $\hat{\sigma}$ for small $s\ll 1$, which is explicitly shown in Appendix \ref{appendix:hamiltonian}.
Meanwhile, since noise may occur in any relevant modes in the system, the quantum state $\hat{\rho}(s)$ is full rank after undergoing excitation noise.

Recall that QFI is written as $H(s)=\text{Tr}[\hat{\rho}(s)\hat{L}(s)^2]$, where $\hat{L}$ is a symmetric logarithmic derivative operator satisfying the equation $\partial_s\hat{\rho}(s)=[\hat{\rho}(s)\hat{L}(s)+\hat{L}(s)\hat{\rho}(s)]/2$ \cite{helstrom1976quantum, holevo2011probabilistic, braunstein1994statistical, paris2009quantum}.
Writing the quantum state in a spectral decomposition form $\hat{\rho}(s)=\sum_i p_i|\psi_i\rangle\langle\psi_i|$ and using $d\hat{\rho}/ds\propto s \hat{\sigma}$, the symmetric logarithmic derivative operator can be written as \cite{paris2009quantum}
\begin{align}
\hat{L}(s)&=2\sum_{i,j:p_i+p_j>0}\frac{\langle \psi_i|\partial_s\hat{\rho}(s)|\psi_j\rangle }{p_i+p_j}|\psi_i\rangle\langle\psi_j| \nonumber \\ 
&\approx 2 s\sum_{i,j:p_i+p_j>0}\frac{\langle \psi_i|\hat{\sigma}|\psi_j\rangle }{p_i+p_j}|\psi_i\rangle\langle\psi_j|+O(s^2).
\end{align}
By the definition of excitation noise, $p_i+p_j>0$ for all $i,j$, and $p_i+p_j$ does not converge to zero as $s\to 0$; hence, $H(s)=\text{Tr}[\hat{\rho}\hat{L}^2]\propto s^2\to0$ as $s\to0$ \cite{finite}.
(A similar argument has been used in the context of quantum spectroscopy \cite{gefen2019overcoming}.)
\end{proof}
Note that although we assumed excitation noise for simplicity, it is sufficient for a final state to be full rank only in the subspace of signal operator $\hat{\sigma}$ to prove the same result.
The proposition can be intuitively explained by noting that the signal in the imaging system approaches zero for $s\to0$, indicated by $d\hat{\rho}/ds\propto s \hat{\sigma}$, while the noise ratio remains finite. Therefore, the SNR vanishes for small $s$, which leads to vanishing QFI.
In contrast, when the quantum state is not full rank in the support of $\hat{\sigma}$ around $s=0$, there exists $p_i>0$ and $|\psi_i\rangle$ such that $p_i\to0$ as $s\to0$ and a projection measurement onto $|\psi_i\rangle\langle\psi_i|$ gives a nonzero QFI element.
Therefore, the QFI may not vanish for a small $s$, which accounts for nonzero QFI for noiseless cases \cite{tsang2016quantum, nair2016far, lupo2016ultimate}.

Note that attenuation channels, where a vacuum state occupies the environmental mode $\hat{e}$, do not introduce false excitation but diminish the signal.
Thus, the QFI of $s$ does not necessarily vanish as $s\to0$.
We emphasize that the proposition does not rule out the possibility of superresolution overcoming the Rayleigh limit
but implies that when the objects are close and the system is noisy, the QFI of the separation can be extremely small.
We supply an important example to analyze the effect of noise in the following section.

\emph{Two identical thermal sources.---} 
Consider two incoherent thermal sources with an unknown separation $s$.
When the modes $\hat{a}_1$, $\hat{a}_2$ are occupied by thermal states of the mean photon number $N_s$, the symmetric and antisymmetric modes $\hat{a}_+$ and $\hat{a}_-$ can also be described by thermal states of the mean photon number $\eta N_s(1+\delta)$ and $\eta N_s(1-\delta)$, respectively \cite{lupo2016ultimate, nair2016far}.
Introducing thermal noise characterized by the same mean photon number $N_n$ onto the relevant modes $\hat{a}_\pm$ and $\hat{b}_\pm$, the quantum state is written as the product of the states of symmetric and antisymmetric modes, $\hat{\rho}=\hat{\rho}_+\otimes\hat{\rho}_-$, where
\begin{align}\label{state}
\hat{\rho}_\pm(s)=\hat{\rho}_\text{T}(\eta N_s(1\pm\delta(s))+N_n)\otimes \hat{\rho}_\text{T}(N_n).
\end{align}
Here, each mode corresponds to $\hat{a}_\pm,\hat{b}_\pm$, respectively, and $\hat{\rho}_\text{T}(N)$ represents a thermal state with the mean photon number $N$.

\begin{figure}[t]
\centering
\includegraphics[width=0.49\textwidth]{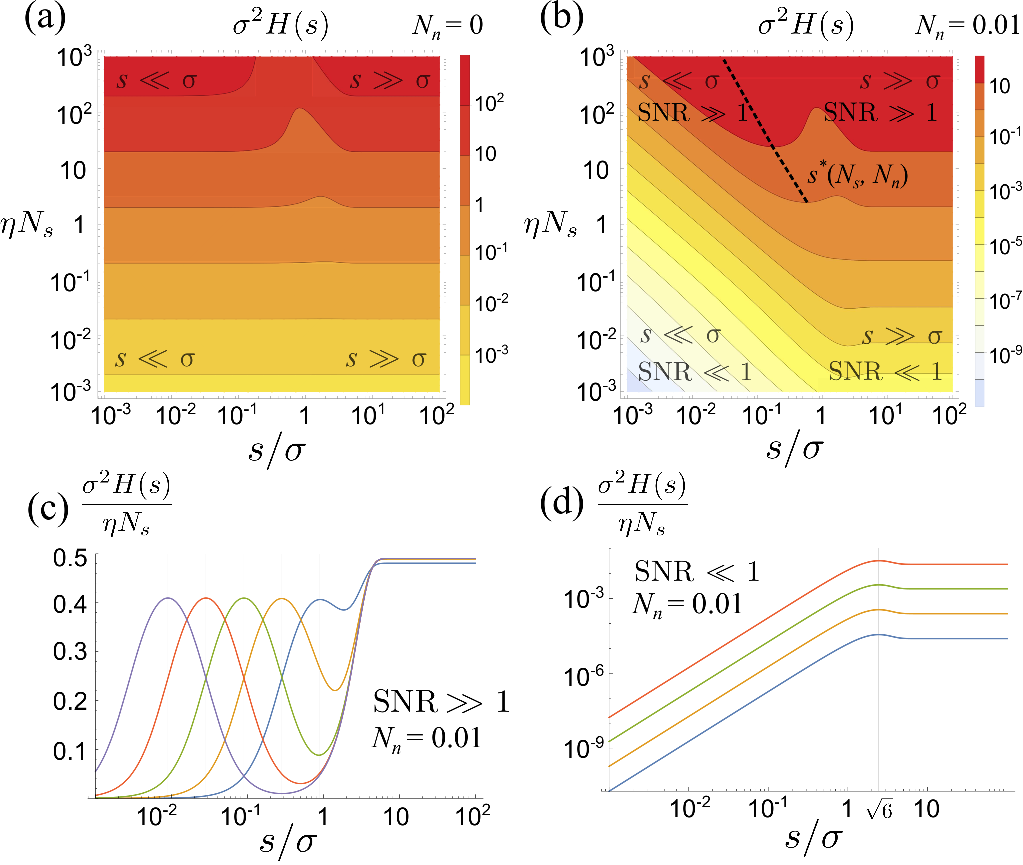}
\caption{QFI with respect to $s$ and $N_s$ with (a) $N_n=0$ (noiseless) and (b) $N_n=0.01$. In the noiseless case, the quantum Fisher information does not decrease as $s$ decreases. However, even with a small amount of noise photons, the QFI drops for small $s$. The dotted line in (b) shows local maxima of QFI for fixed $\eta N_s$, $N_n>0$, and $\text{SNR}\gg 1$ as shown in (c). (c) Normalized QFI when $\text{SNR}\gg 1$ with respect to $s$ with $\eta N_s=10^4, 10^3, 10^2, 10, 1$ from the left to the right and $N_n=0.01$. The horizontal line represents $H(s^*)$ and the vertical lines $s^*$ (see the main text). It captures the nonmonotonic behavior of QFI. (d) Normalized QFI when $\text{SNR}\ll 1$ with $\eta N_s=10^{-4}, 10^{-3}, 10^{-2}, 10^{-1}$  from the bottom.}
\label{qfi_thermal}
\end{figure}

Using the QFI formula of Gaussian states \cite{pinel2013quantum, jiang2014quantum, vsafranek2015quantum, vsafranek2016optimal, nichols2018multiparameter, vsafranek2018estimation, oh2019optimal}, we obtained the QFI of the separation $s$ (See Appendix \ref{appendix:QFI}), $H(s)=H_+(s)+H_-(s)$ with
\begin{align}\label{qfi}
H_\pm(s)&=\frac{\eta^2N_s^2\gamma^2}{(\eta N_s(1\pm\delta)+N_n+1)(\eta N_s(1\pm\delta)+N_n)} \nonumber \\ 
&-\frac{2\eta^2N_s^2[(1\pm\delta)(\delta''(0)\mp\delta''(s))+\gamma^2]}{(2N_n+1)(2\eta N_s(1\pm\delta)+2N_n+1)-1}.
\end{align}
Here, $H_\pm(s)$ represent the QFI from symmetric and antisymmetric modes, respectively.
The first and second term accounts for the changes of the mean photon number on mode $\hat{a}_\pm$ from the change of effective attenuation factors $\eta_\pm$ and the transformation of the spatial modes' shape $\hat{a}_\pm(s)$ into $\hat{a}_\pm(s+ds)\approx \hat{a}_\pm(s)+ds \partial_s \hat{a}_\pm$, respectively.

The QFI recovers previous results when $N_n=0$ in Refs. \cite{lupo2016ultimate, nair2016far}.
More importantly, the QFI vanishes as $s\to0$ unless $N_n=0$.
Figure~\ref{qfi_thermal} (a) and (b) compare the QFI $H(s)$ in the ideal and noisy cases with the Gaussian PSF, $\psi(x)=e^{-x^2/4\sigma^2}/(2\pi\sigma^2)^{1/4}$.
A remarkable difference between the two cases is that as $s\to0$, the QFI in the noisy case rapidly drops whereas it does not change in the ideal case.
For example, when the separation $s$ is $0.01\sigma$ and the mean signal photons $\eta N_s$ is $1$, the QFI $H(s)$ is $0.5/\sigma^2$ and $6\times 10^{-4}/\sigma^2$ for the noiseless case and the noisy case with $N_n=0.01$, respectively, which clearly shows that even a small amount of noise can be critical to the resolution.

Let us consider the regime where the SNR is large, $\text{SNR}\equiv \eta N_s/N_n\gg 1$.
In this regime, Fig. \ref{qfi_thermal} (c) shows another interesting feature of QFI; it is not monotonic with respect to $s$.
For a small separation $s \ll \sigma$ in the regime, the QFI for the Gaussian PSF can be approximated by
\begin{align}
H(s)\approx \frac{4\eta^2 N_s^2 s^2}{\eta^2 N_s^2 s^4+8\eta N_s s^2 \sigma^2+64N_n(N_n+1)\sigma^4},
\end{align}
which has the local maximum 
\begin{align}
H(s^*)&~~\approx \frac{\eta N_s}{2\sigma^2}\frac{\sqrt{N_n^2+N_n}}{(N_n+\sqrt{N_n^2+N_n})(\sqrt{N_n^2+N_n}+N_n+1)} \nonumber \\ 
&\stackrel{N_n\ll 1}{\approx} \frac{\eta N_s}{2\sigma^2}\frac{1}{1+2\sqrt{N_n}}
\end{align}
at $s^*=2\sqrt{2}(N_n^2+N_n)^{1/4}\sigma/\sqrt{\eta N_s}$, as shown in Fig. \ref{qfi_thermal} (c).
Here $s^*$ is a characteristic length scale, and if $s\ll s^*$, the QFI can be further approximated as
\begin{align}
H(s)&\approx\frac{\eta^2 N_s^2}{N_n(N_n+1)}\Delta k^4 s^2= \frac{\eta^2 N_s^2}{N_n(N_n+1)}\frac{s^2}{16 \sigma^4}  \\ 
&\stackrel{N_n\ll 1}{\approx}\frac{\eta^2 N_s^2}{N_n}\frac{s^2}{16 \sigma^4}  ~~~~\text{if}~~ \eta N_s\gg N_n ~\text{and}~ s \ll s^*. \nonumber
\end{align}
One can observe that when $\text{SNR}\gg 1$, $N_n\ll 1$ and $s\ll s^*$, the QFI per a signal photon is proportional to the SNR $H(s)/\eta N_s\propto \eta N_s/N_n$, which is consistent with the previous results \cite{len2020resolution, lupo2020subwavelength}.
Also, the QFI decreases quadratically as $s\to 0$.

On the other hand, when the SNR is small, i.e., $\eta N_s/N_n \ll 1$, and the separation is small, $s\ll \sqrt{6}\sigma$,
the QFI is approximated by
\begin{align}
H(s)&\approx \frac{\eta^2 N_s^2}{2N_n(N_n+1)}[3\Delta k^4+\delta^{(4)}(0)]s^2 \\ 
&= \frac{\eta^2 N_s^2}{N_n(N_n+1)}\frac{3s^2}{16} ~~~\text{if}~ \eta N_s\ll N_n \text{ and } s\ll\sqrt{6} \sigma, \nonumber
\end{align}
which is shown in Fig. \ref{qfi_thermal} (d).
Again, when $N_n\ll 1$, the QFI per a signal photon is proportional to the SNR, $H(s)/\eta N_s\propto \eta N_s/N_n,$ and decreases quadratically as $s\to 0$.

Finally, for a large separation $s\gg \sigma$, the QFI can be approximated as $H(s)\approx 2\eta^2 N_s^2\Delta k^2/[2N_n^2+\eta N_s+2N_n(\eta N_s+1)]$, which shows that the noise decreases the QFI for a large separation as well.

As a remark, we compare the QFI in Eq.~\eqref{qfi} with the one obtained in Ref.~\cite{lupo2020subwavelength} where the same type of noise was studied in the imaging process.
The discrepancy of the expression is present because the noise model of Ref.~\cite{lupo2020subwavelength} assumes that noise occurs only on the modes $\hat{a}_\pm$ whereas our noise model assumes the same amount of noise on $\hat{b}_\pm$ modes.
Nevertheless, the previous result has also revealed that the QFI vanishes as $s\to0$ because the rank of the quantum state does not change around $s=0$ in the first order of $s$ even if we assume $N_n=0$ for $\hat{b}_\pm$ modes.

\emph{Noisy detectors.---}
As pointed out in Ref.~\cite{lupo2020subwavelength}, analyzing QFI might not be appropriate for considering the effect of dark counts because QFI is a measurement-independent quantity while dark counts are a feature of the measurement device.
To analyze the impact of dark counts, we employ CFI $F(\theta)$ for an unknown parameter $\theta$, the inverse of which gives a lower bound of estimation error for a given measurement apparatus $\Delta^2\theta \geq 1/MF(\theta)$ \cite{rao1992information, kay1993fundamentals, cramer1999mathematical, van2004detection}.
By introducing the following proposition, we show that excitation noise on detectors makes the CFI vanish.
\begin{proposition}
Consider a quantum state that satisfies $\partial_\theta\hat{\rho}\approx \theta\hat{\sigma}$ for small $\theta$ with a Hermitian operator $\hat{\sigma}$ and a positive-operator-valued-measurement (POVM) $\{\hat{\Pi}_k\}_{k\in K}$ satisfying $\hat{\Pi}_k\geq 0$, $\sum_{k \in K}\hat{\Pi}_k=\mathbb{1}$.
Here, $K$ is an index set of measurement outcomes.
If the support of the probability distribution $p_k=\text{Tr}[\hat{\rho}(\theta)\hat{\Pi}_k]$, $\{k\in K|p_k>0\}$ does not change as $\theta\to0$, the CFI converges to zero as $\theta\to0$.
\end{proposition}
\begin{proof}
Recall that the CFI of probability distribution $\{p_k\}$ is given by \cite{rao1992information, kay1993fundamentals, cramer1999mathematical, van2004detection}
\begin{align}
F(\theta)=\sum_{p_k>0} \frac{1}{p_k}\left(\frac{\partial p_k}{\partial \theta}\right)^2.
\end{align}
The probability of obtaining outcome $k$ by measuring a quantum state $\hat{\rho}(\theta)$ with POVM $\{\hat{\Pi}_k\}_{k\in K}$ and its derivative with respect to $\theta$ are given by
\begin{align}
p_k=\text{Tr}[\hat{\Pi}_k \hat{\rho}(\theta)]~~~ \text{and} ~~~ \frac{\partial p_k}{\partial \theta}\approx \theta\text{Tr}[\hat{\Pi}_k \hat{\sigma}].
\end{align}
Therefore, the CFI of small $\theta$ is written as
\begin{align}
F(\theta)=\sum_{p_k>0} \frac{1}{p_k}\left(\frac{\partial p_k}{\partial \theta}\right)^2\approx \theta^2 \sum_{p_k>0}\frac{1}{p_k}\left(\text{Tr}[\hat{\Pi}_k \hat{\sigma}]\right)^2.
\end{align}
Similar to the QFI, the CFI converges to zero as $\theta\to0$ unless there exists $p_k$ such that $p_k\to0$ \cite{finite}.
\end{proof}
Proposition 2 can be understood similarly to Proposition 1.
The proposition provides a necessary condition to prevent the CFI of a measurement setting from vanishing for a small separation $s$.
For example, dark counts are a kind of excitation noise on detectors that causes false excitations on all relevant detectors.
Dark count rates are generally nonzero in all relevant detectors in practice; thus, the support of the probability distribution does not change, and it is natural to expect that the CFI of separation $s$ vanishes $F(s)\to0$ as $s\to 0$ in experiment.
Moreover, the proposition can be applied to measurement crosstalk, which may arise for fin-SPADE scheme \cite{gessner2020superresolution}.
It makes all measurement outcomes mixed so that eventually the probability of obtaining each outcome becomes nonzero.
Also, the proposition shows the limitation of direct imaging, homodyne detection, and heterodyne detection \cite{yang2017fisher} which always give a nonzero probability of each outcome for generic PSFs even in the noiseless case.
As a final remark, Proposition 2 does not imply the failure of superresolution; it suggests that excitation noise on detectors can pose a limit on the resolution as for QFI in the previous section.

\begin{figure}[t!]
\centering
\includegraphics[width=0.35\textwidth]{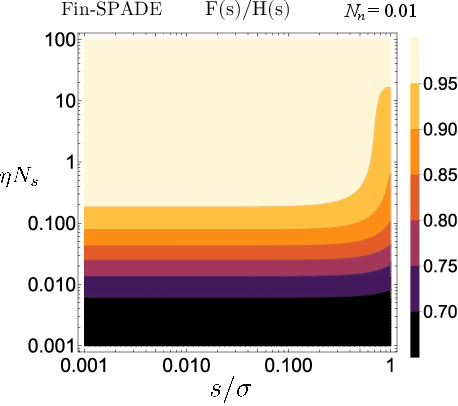}
\caption{Relative CFI of fin-SPADE to QFI with respect to different separation $s$ and mean signal photon number $\eta N_s$.}
\label{cfi}
\end{figure}

\emph{Finite spatial mode demultiplexing.---}
Finally, we analyze an achievable resolution using the fin-SPADE measurement.
In the noiseless case, the fin-SPADE scheme employs a photon counting for each Hermite-Gaussian mode $h_q(x)$ on the image plane, which is optimal if enough Hermite-Gaussian modes are accessible in experiment \cite{tsang2016quantum, lupo2016ultimate}.
In general, the analytical expression of the CFI of the fin-SPADE scheme is difficult to obtain due to the statistical correlations between different modes of the measurement.
We thus obtain the lower bound of the CFI using an inequality $F(\theta)\geq \dot{\vec{\mu}}^\text{T}C^{-1}\dot{\vec{\mu}}$, where $\vec{\mu}$ and $C$ denote the mean and covariance matrix of the outcome distribution, and $\dot{\vec{\mu}}\equiv\partial_s \vec{\mu}$ \cite{stein2014lower}.
We provide more details on the CFI and the numerical method in Appendix \ref{appendix:CFI}.
We consider a finite number of Hermite-Gaussian modes $h_q$ with $0\leq q\leq Q-1$ with $Q=15$ in the presence of thermal noise in the problem of resolving two incoherent thermal sources. 
We confirmed that increasing $Q$ larger than $15$ does not change the CFI for $10^{-3}\leq s/\sigma \leq 1$.
Figure \ref{cfi} shows the ratio of the lower bound of the CFI of fin-SPADE to the QFI (See Appendix \ref{appendix:CFI}).
It clearly shows that for a large number of signal photons $\eta N_s$, the ratio converges to the unity, which indicates that the fin-SPADE measurement is optimal in that regime.
Even when $\eta N_s$ is small, the lower bound of the CFI gives at least $65\%$ of the QFI.
Hence, the fin-SPADE method's performance is not degraded significantly by thermal noise compared to the QFI.
A particular way to improve this further is to directly measure the incoming photon numbers onto the symmetric and antisymmetric modes and their derivative modes $\{\hat{a}_\pm,\hat{b}_\pm \}$ (See Appendix \ref{appendix:QFI}).
In general, the implementation of such a measurement requires a prior information, which might be overcome by using the adaptive method \cite{barndorff2000fisher}.

\emph{Conclusions and discussion.---}
In this Letter, we have investigated the effect of noise on the resolution of two identical sources with an arbitrary state using quantum and classical Fisher information and shown that the Fisher information converges to zero if the system suffers from false excitation noise such as thermal noise or dark counts.
We have shown that in the problem of resolving two incoherent thermal sources with the number of signal photons being larger than that of noise photons, a signal-to-noise ratio poses a fundamental limit.
Finally, we have shown that the finite spatial demultiplexing measurement is nearly optimal for a large signal-to-noise ratio.

Throughout the Letter, we are assuming that two objects are identical.
Thus, the same conclusion might not hold if the sources are not identical \cite{vrehavcek2017multiparameter, vrehavcek2018optimal, bonsma2019realistic, prasad2020quantum}.
It would be interesting to analyze the problem of resolving nonidentical sources.

We acknowledge useful discussions with Cosmo Lupo.
We acknowledge support from the ARL-CDQI (W911NF-15-2-0067), ARO (W911NF-18-1-0020, W911NF-18-1-0212), ARO MURI (W911NF-16-1-0349), AFOSR MURI (FA9550-15-1-0015, FA9550-19-1-0399), NSF (EFMA-1640959, OMA-1936118, EEC-1941583), NTT Research, and the Packard Foundation (2013-39273).

\appendix
\begin{widetext}
\section{Dynamics of imaging process}\label{appendix:hamiltonian}
In this section, we provide the details of the imaging process and show that the quantum state of imaging process can be linearized by the separation $s$ when $s\ll1$.
We follow the method introduced in Ref. \cite{lupo2016ultimate}.
As introduced in the main text, two identical sources with separation $s>0$ emit light that excites modes characterized by $\hat{c}_{1,2}^\dagger$, and the emitted light is attenuated and distorted when it arrives at the image plane such that
\begin{align}
\hat{a}_{1,2}^\dagger=\sqrt{\eta}\hat{c}_{1,2}^\dagger+\sqrt{1-\eta}\hat{v}_{1,2}^\dagger,
\end{align}
where $\hat{v}_{1,2}$ represent the environmental mode operators interacting with $\hat{c}_{1,2}$, respectively.
Introducing the symmetric and antisymmetric mode operators, 
\begin{align}
\hat{a}_\pm=\frac{\hat{a}_1\pm\hat{a}_2}{\sqrt{2(1\pm\delta(s))}}, ~~~~\hat{c}_\pm=\frac{\hat{c}_1\pm\hat{c}_2}{\sqrt{2}},
\end{align}
and inverting Eq.~(3) in the main text,
we write
\begin{align}
\hat{a}_\pm=\sqrt{\eta_\pm}\hat{c}_\pm+\sqrt{1-\eta_\pm}\hat{v}_\pm=e^{i\hat{H}_\pm\theta_\pm}\hat{c}_\pm e^{-i\hat{H}_\pm\theta_\pm},
\end{align}
with $\eta_\pm\equiv \eta(1\pm\delta)$, $\hat{H}_\pm=i(\hat{c}_\pm^\dagger \hat{v}_\pm-\hat{v}_\pm^\dagger \hat{c}_\pm)$, and $\theta_\pm\equiv \arccos\sqrt{\eta_\pm}$.
Thus, when the separation is $s$, the quantum state on the image plane is written as 
\begin{align}\label{s_eq}
\hat{\rho}(s)=\text{Tr}_{u_\pm}\left[e^{-i(\hat{H}_+\theta_++\hat{H}_-\theta_-)}\left(\hat{\rho}_{c_+c_-}\otimes \hat{\sigma}_{v_+v_-}\right)e^{i(\hat{H}_+\theta_++\hat{H}_-\theta_-)}\right],
\end{align}
where $\hat{\rho}_{c_+c_-}$ represents the quantum state of light emitted by the sources, and $\hat{\sigma}_{v_+v_-}$ represents the quantum state of the environment. 
When $s$ infinitesimally changes, the quantum state can be written as
\begin{align}\label{ds_eq}
\hat{\rho}(s+ds)=\text{Tr}_{u_\pm}\left[e^{-i(\hat{H}_+\tilde{\theta}_++\hat{H}_-\tilde{\theta}_-)}\left(\hat{\rho}_{c_+c_-}\otimes \hat{\sigma}_{v_+v_-}\right)e^{i(\hat{H}_+\tilde{\theta}_++\hat{H}_-\tilde{\theta}_-)}\right]
\end{align}
Here, $\tilde{\theta}=\theta(s+ds)$. Notice that the quantum state is written in $\hat{a}_\pm(s+ds)$ modes.
In order to write the quantum state in terms of $\hat{a}_\pm(s)$ modes as Eq.~\eqref{s_eq}, we describe the dynamics of the mode operators $\hat{a}_\pm$.
Using the Heisenberg equation of motion, we obtain
\begin{align}\label{h_eq}
\frac{d\hat{a}_\pm}{ds}=i\frac{d\theta_\pm}{ds}[\hat{H}_\pm,\hat{a}_\pm]+\frac{\partial \hat{a}_\pm}{\partial s}=i\frac{d\theta_\pm}{ds}[\hat{H}_\pm,\hat{a}_\pm]-\frac{\epsilon_\pm}{2\sqrt{1\pm\delta}}\hat{b}_\pm=i\left[\frac{d\theta_\pm}{ds}\hat{H}_\pm+i\frac{\epsilon_\pm}{2\sqrt{1\pm\delta}}(\hat{a}_\pm\hat{b}_\pm^\dagger-\hat{a}_\pm^\dagger\hat{b}_\pm),\hat{a}_\pm\right]\equiv i[\hat{H}_\pm^\text{eff},\hat{a}_\pm].
\end{align}

So far, we have reproduced Lemma 1 in Ref. \cite{lupo2016ultimate}. From now on, we analyze the dynamics of the system and show that the derivative of the quantum state with respect to the separation $s$ is linearized in $s$ for a small $s$ limit to apply proposition 1 in the main text.
Defining $\hat{\gamma}\equiv e^{-i(\hat{H}_+\theta_++\hat{H}_-\theta_-)}\left(\hat{\rho}_{c_+c_-}\otimes\hat{\sigma}_{v_+v_-}\right)e^{i(\hat{H}_+\theta_++\hat{H}_-\theta_-)}\otimes|0\rangle\langle 0|_{b_+}\otimes|0\rangle\langle 0|_{b_-}$ and $B_\pm\equiv -{\epsilon_\pm}/(2\sqrt{1\pm\delta})$ and using the equation of motion Eq.~\eqref{h_eq} and Eq.~\eqref{s_eq}, Eq.~\eqref{ds_eq} can be equivalently written as, 
\begin{align}
\hat{\rho}(s+ds)&\approx\text{Tr}_{u_\pm}\left[e^{-i(\hat{H}^\text{eff}_++\hat{H}^\text{eff}_-)ds}\hat{\gamma}e^{i(\hat{H}^\text{eff}_++\hat{H}^\text{eff}_-)ds}\right] \nonumber \\ 
&\approx e^{-ds[B_+(\hat{a}_+\hat{b}^\dagger_+-\hat{a}^\dagger_+\hat{b}_+)+B_-(\hat{a}_-\hat{b}^\dagger_--\hat{a}^\dagger_-\hat{b}_-)]}\text{Tr}_{u_\pm}[e^{-i(\hat{H}_+\tilde{\theta}_++\hat{H}_-\tilde{\theta}_-)}\left(\hat{\rho}_{c_+c_-}\otimes\hat{\sigma}_{v_+v_-}\otimes|0\rangle\langle 0|_{b_+}\otimes|0\rangle\langle 0|_{b_-}\right)e^{i(\hat{H}_+\tilde{\theta}_++\hat{H}_-\tilde{\theta}_-)}] \nonumber \\
&~~~~~\times e^{ds[B_+(\hat{a}_+\hat{b}^\dagger_+-\hat{a}^\dagger_+\hat{b}_+)+B_-(\hat{a}_-\hat{b}^\dagger_--\hat{a}^\dagger_-\hat{b}_-)]} \nonumber \\
&\approx \left[1-ds[B_+(\hat{a}_+\hat{b}^\dagger_+-\hat{a}^\dagger_+\hat{b}_+)+B_-(\hat{a}_-\hat{b}^\dagger_--\hat{a}^\dagger_-\hat{b}_-)]\right] \nonumber\\
&~~~~~\times\text{Tr}_{u_\pm}[(1-ids(\hat{H}_+\partial_s\theta_++\hat{H}_-\partial_s\theta_-))\hat{\gamma}(1+ids(\hat{H}_+\partial_s\theta_++\hat{H}_-\partial_s\theta_-))] \nonumber \\ 
&~~~~~\times\left[1+ds[B_+(\hat{a}_+\hat{b}^\dagger_+-\hat{a}^\dagger_+\hat{b}_+)+B_-(\hat{a}_-\hat{b}^\dagger_--\hat{a}^\dagger_-\hat{b}_-)]\right] \nonumber \\ 
&=\hat{\rho}(s)-ds[B_+(\hat{a}_+\hat{b}^\dagger_+-\hat{a}^\dagger_+\hat{b}_+)+B_-(\hat{a}_-\hat{b}^\dagger_--\hat{a}^\dagger_-\hat{b}_-),\hat{\rho}(s)]-ids\text{Tr}_{u_\pm}\left([\hat{H}_+\partial_s\theta_++\hat{H}_-\partial_s\theta_-,\hat{\gamma}]\right).
\end{align}
Thus, the derivative of the quantum state is written as
\begin{align}\label{deriv_rho}
\frac{d\hat{\rho}(s)}{ds}&=-[B_+(\hat{a}_+\hat{b}^\dagger_+-\hat{a}^\dagger_+\hat{b}_+)+B_-(\hat{a}_-\hat{b}^\dagger_--\hat{a}^\dagger_-\hat{b}_-),\hat{\rho}(s)]-i\text{Tr}_{u_\pm}\left(\left[\hat{H}_+\partial_s\theta_++\hat{H}_-\partial_s\theta_-,\hat{\gamma}\right]\right).
\end{align}
Now, let us consider a regime where the separation $s$ is small.
For small $s$, we can approximate
\begin{align}
B_+&\approx -\frac{1}{4}\sqrt{(\delta^{(4)}(0)-\delta''(0)^2)}s+O(s^2)\propto \alpha_+ s+O(s^2), \\
B_-&\approx -\sqrt{\frac{1}{12\delta''(0)}\left(\frac{\delta^{(6)}(0)}{5}-\frac{\delta^{(4)}(0)^2}{3\delta''(0)}\right)}s+O(s^2)\propto \alpha_- s+O(s^2),
\end{align}
where $\delta^{(n)}(0)\equiv\partial^n\delta(s)/\partial s^n|_{s=0}$.
Thus, the first commutator in Eq.~\eqref{deriv_rho} is linearized in $s$ for small $s$.
Now, let us focus on the second term.
Let us assume that $\hat{\sigma}_{v_+v_-}=\hat{\sigma}_{v_+}\otimes \hat{\sigma}_{v_-}$ which is a natural choice as a quantum state for environment. Note that the quantum state is written in $\hat{a}_\pm(s)$ modes.
For small $s$, noting that
\begin{align}
\frac{d\theta_+}{ds}&\approx -\frac{\eta \delta'(s)}{2\sqrt{\eta(1+\delta(s))}\sqrt{1-\eta(1+\delta(s))}}\approx -\frac{\sqrt{\eta} \delta''(s)s}{\sqrt{8(1-2\eta)}}+\mathcal{O}(s^2)\\
\frac{d\theta_-}{ds}&\approx \frac{\eta \delta'(s)}{2\sqrt{\eta(1-\delta(s))}\sqrt{1-\eta(1-\delta(s))}}\approx \frac{\sqrt{\eta} \delta'(s)}{2\sqrt{1-\delta(s)}}\approx -\sqrt{-\frac{\eta \delta''(0)}{2}}+\mathcal{O}(s),
\end{align}
we have
\begin{align}
&\text{Tr}_{v_\pm}\left([\hat{H}_+\partial_s\theta_+,\hat{\gamma}]\right)\propto s+O(s^2).
\end{align}

On the other hand, we can expand the remaining term in $s$ around $s=0$ as
\begin{align}
&\text{Tr}_{u_\pm}\left([\hat{H}_-\partial_s\theta_-,e^{-i(\hat{H}_+\theta_++\hat{H}_-\theta_-)}\left(\hat{\rho}_{c_+c_-}\otimes\hat{\sigma}_{v_+}\otimes\hat{\sigma}_{v_-}\right)e^{i(\hat{H}_+\theta_++\hat{H}_-\theta_-)}]\right) \nonumber \\
=&\frac{d\theta_-}{ds}\text{Tr}_{u_+}\left(e^{-i\hat{H}_+\theta_+}\left(\text{Tr}_{u_-}\left[[\hat{H}_-,e^{-i\hat{H}_-\theta_-}\left(\hat{\rho}_{c_+c_-}\otimes\hat{\sigma}_{v_-}\right)e^{i\hat{H}_-\theta_-}]\right]\otimes\hat{\sigma}_{v_+}\right)e^{i\hat{H}_+\theta_+}\right) \nonumber \\
\approx&\frac{d\theta_-}{ds}\text{Tr}_{u_+}\left(e^{-i\hat{H}_+\theta_+}\left(\text{Tr}_{u_-}\left[[\hat{H}_-,(1-is\hat{H}_-\partial_s\theta_-)e^{-i\pi\hat{H}_-/2}\left(\hat{\rho}_{c_+c_-}\otimes\hat{\sigma}_{v_-}\right)e^{i\pi\hat{H}_-/2}(1+is\hat{H}_-\partial_s\theta_-)]\right]\otimes\hat{\sigma}_{v_+}\right)e^{i\hat{H}_+\theta_+}\right) \nonumber \\
\approx&\frac{d\theta_-}{ds}\text{Tr}_{u_+}\left(e^{-i\hat{H}_+\theta_+}\left(\text{Tr}_{u_-}\left[[\hat{H}_-,(1-is\hat{H}_-\partial_s\theta_-)\left(\hat{\rho}_{c_+v_-}\otimes\hat{\sigma}_{c_-}\right)(1+is\hat{H}_-\partial_s\theta_-)]\right]\otimes\hat{\sigma}_{v_+}\right)e^{i\hat{H}_+\theta_+}\right) \nonumber \\
\approx&\frac{d\theta_-}{ds}\text{Tr}_{u_+}\left(e^{-i\hat{H}_+\theta_+}\left(\text{Tr}_{u_-}\left[[\hat{H}_-,\hat{\rho}_{c_+v_-}\otimes\hat{\sigma}_{c_-}]\right]\otimes\hat{\sigma}_{v_+}\right)e^{i\hat{H}_+\theta_+}\right)+\mathcal{O}(s).
\end{align}
We have used the fact that $\theta_-\to \pi/2$ as $s\to0$ to expand the unitary operator $e^{-i\hat{H}_-\theta_-}$.
Thus, the zeroth order of $s$ becomes zero if $\text{Tr}_{u_-}\left([\hat{H}_-,\hat{\rho}_{c_+u_-}\otimes\hat{\sigma}_{c_-}]\right)=0$.
The condition becomes
\begin{align}
&\text{Tr}_{u_-}\left([\hat{H}_-,\hat{\rho}_{c_+u_-}\otimes\hat{\sigma}_{c_-}]\right)=i\text{Tr}_{u_-}\left([\hat{c}_-^\dagger \hat{v}_--\hat{v}^\dagger_-\hat{c}_-,\hat{\rho}_{c_+u_-}\otimes\hat{\sigma}_{c_-}]\right) \nonumber \\ 
=&i\text{Tr}_{u_-}(\hat{\rho}_{c_+v_-}\hat{v}_-)[\hat{c}_-^\dagger,\hat{\sigma}_{c_-}]+i\text{Tr}_{u_-}(\hat{\rho}_{c_+v_-}\hat{v}_-^\dagger)[\hat{\sigma}_{c_-}^\dagger,\hat{c}_-]=0.
\end{align}
Thus, this condition is satisfied if $\text{Tr}_{u_-}(\hat{\rho}_{c_+v_-}\hat{v}_-)=0$.
Since we assume two identical objects, the quantum state $\hat{\rho}_{c_+c_-}$ satisfies 
\begin{align}
\text{Tr}_{c_-}[\hat{\rho}_{c_+c_-}\hat{c}_-]=\text{Tr}_{c_-}[\hat{\rho}_{c_+c_-}\frac{\hat{c}_1-\hat{c}_2}{\sqrt{2}}]=\text{Tr}_{c_-}[\hat{\rho}_{c_+c_-}\frac{\hat{c}_2-\hat{c}_1}{\sqrt{2}}]. 
\end{align}
Thus, $\text{Tr}_{u_-}[\hat{\rho}_{c_+v_-}\hat{v}_-]=0$.
Hence, we have shown that $\partial\hat{\rho}/\partial s\propto s+O(s^2)$.

\section{Quantum Fisher information for two incoherent thermal sources}\label{appendix:QFI}
In this section, we derive quantum Fisher information of the separation of two identical thermal sources and obtain the optimal measurement corresponding to the quantum Fisher information.
Quantum Fisher information of $n$-mode Gaussian states is well-known, which is given by \cite{pinel2013quantum, jiang2014quantum, vsafranek2015quantum, vsafranek2016optimal, nichols2018multiparameter, vsafranek2018estimation, oh2019optimal}
\begin{align}
H(s)=-\text{Tr}[G \frac{\partial V(s)}{\partial s}],
\end{align}
where $V(s)$ is the $2n\times 2n$ covariance matrix of the Gaussian state $\hat{\rho}$, $V_{ij}=\text{Tr}[\hat{\rho}\{\hat{Q}_i-\langle\hat{Q}_i\rangle,\hat{Q}_j-\langle\hat{Q}_j\rangle\}]/2$, $\hat{Q}=(\hat{x}_1,\hat{p}_1,\cdots, \hat{x}_n, \hat{p}_n)^\text{T}$ and $\Omega$ a skew symmetric matrix giving the canonical commutation relation, 
\begin{align}
\Omega=\mathbb{1}_{n}\otimes \omega, ~~~ \omega=
\begin{pmatrix}
0 & 1 \\
-1 & 0
\end{pmatrix},
\end{align}
and $G$ is a $2n \times 2n$ real symmetric matrix satisfying
\begin{align}\label{G_eq}
4V(s)GV(s)+\Omega G \Omega+2\frac{\partial V(s)}{\partial s}=0.
\end{align}
Here, $\mathbb{1}_n$ denotes the $n \times n$ identity matrix.

When two incoherent sources of a distance $s$ are in thermal states with a same temperature characterized by the mean photon number $N_s$, the quantum state can be written as a product form of states in modes $\hat{c}_\pm$, $\hat{\rho}_\text{T}(N_s)\otimes \hat{\rho}_\text{T}(N_s)$.
When the light arrived at the image plane, the quantum state is described in symmetric and antisymmetric modes as, 
\begin{align}
\hat{\rho}(s)&=\text{Tr}_{v_+v_-}\left[e^{-i(\hat{H}_+\theta_++\hat{H}_-\theta_-)}\left[\hat{\rho}_\text{T}(N_s)\otimes \hat{\rho}_\text{T}(N_s)\otimes |0\rangle\langle0|\otimes |0\rangle\langle 0|\right]_{c_+c_-v_+v_-}e^{i(\hat{H}_+\theta_++\hat{H}_-\theta_-)}\otimes |0\rangle\langle0|_{b_+}\otimes |0\rangle\langle 0|_{b_-} \right] \nonumber \\ 
&=[\hat{\rho}_\text{T}(\eta_+N_s)\otimes \hat{\rho}_\text{T}(\eta_-N_s)]_{a_+a_-}\otimes |0\rangle\langle0|_{b_+}\otimes |0\rangle\langle 0|_{b_-} \nonumber \\ 
&=[\hat{\rho}_\text{T}(\eta N_s(1+\delta))\otimes \hat{\rho}_\text{T}(\eta N_s(1-\delta))]_{a_+a_-}\otimes |0\rangle\langle0|_{b_+}\otimes |0\rangle\langle 0|_{b_-}.
\end{align}
Let us introduce a thermal noise assuming that the thermal photon number of the noise is the same on the relevant modes. Thus, the thermal photon numbers on each mode increase as
\begin{align}
\hat{\rho}(s)=[\hat{\rho}_\text{T}(\eta N_s(1+\delta)+N_n)\otimes \hat{\rho}_\text{T}(\eta N_s(1-\delta)+N_n)]_{a_+a_-}\otimes[\hat{\rho}_\text{T}(N_n)\otimes \hat{\rho}_\text{T}(N_n)]_{b_+b_-}.
\end{align}

Let us first focus on the symmetric modes $\hat{a}_+,\hat{b}_+$.
For infinitesimal change of $s$, 
The quantum state in the symmetric modes can be written as
\begin{align}
\hat{\rho}_\pm(s)&=\hat{\rho}_\text{T}(\eta N_s(1\pm\delta)+N_n)\otimes \hat{\rho}_\text{T}(N_n).
\end{align}

The quantum state with an infinitesimal change $ds$ of $s$ is given by
\begin{align}
\hat{\rho}_+(s+ds)&\approx\text{Tr}_{u_+u_-}\left[e^{-i\hat{H}^\text{eff}_+ds}e^{-i\hat{H}_+\theta_+}\left[\hat{\rho}_\text{T}(N_s)\otimes |0\rangle\langle 0|\right]_{c_+v_+}e^{i\hat{H}_+\theta_+}e^{i\hat{H}^\text{eff}_+ds}\right] \nonumber \\ 
&\approx e^{-B_+ds(\hat{a}^\dagger_+\hat{b}_+-\hat{a}_+\hat{b}_+^\dagger)}\left[\hat{\rho}_\text{T}(\tilde{\eta}_+N_s)\otimes |0\rangle\langle0|\right]_{a_+b_+}e^{B_+ds(\hat{a}^\dagger_+\hat{b}_+-\hat{a}_+\hat{b}_+^\dagger)},
\end{align}
where $\tilde{\eta}=\eta[1+\delta(s+ds)]$.
Again, introducing the thermal noise, the state becomes
\begin{align}
\hat{\rho}_+(s+ds)\approx e^{-B_+ds(\hat{a}_+\hat{b}^\dagger_+-\hat{a}^\dagger_+\hat{b}_+)}\left[\hat{\rho}_\text{T}(\eta N_s[1+\delta(s+ds)]+N_n)\otimes \hat{\rho}_\text{T}(N_n)\right]_{a_+b_+}e^{B_+ds(\hat{a}_+\hat{b}^\dagger_+-\hat{a}^\dagger_+\hat{b}_+)}.
\end{align}

Thus, the covariance matrix of the symmetric modes can be written as
\begin{align}
V_+(s+ds)&=SV_+(s)S^\text{T}=
\begin{pmatrix}
\mu_+^2v_1+(1-\mu_+^2)v_2 & -\mu_+\sqrt{1-\mu_+^2}(v_2-v_1) \\ 
-\mu_+\sqrt{1-\mu_+^2}(v_2-v_1) & \mu_+^2v_2+(1-\mu_+^2)v_1
\end{pmatrix}
\otimes  \mathbb{1}_2,
\\ 
~~~ V_+(s)&=\text{diag}(v_1,v_1,v_2,v_2),~~~ S= 
\begin{pmatrix}
\mu_+ & \sqrt{1-\mu_+^2} \\
-\sqrt{1-\mu_+^2} & \mu_+
\end{pmatrix}
\otimes\mathbb{1}
, ~~~ \mu_+=\cos B_+ds.
\end{align}
Here, the first (second) row and column of the first matrix represents the mode $\hat{a}_+$ ($\hat{b}_+$), $v_1=(1+\delta(s+ds))\eta N_s+N_n+1/2$, and $v_2=N_n+1/2$, and $\mu_+$ transmittance of the beam splitter unitary operator.
Noting that 
\begin{align}
\mu_+\simeq 1-\frac{1}{2} B_+(s)^2 ds^2=1+ds^2\left(\frac{\delta''(0)-\delta''(s)}{8(1+\delta)}+\frac{\delta'(s)^2}{8(1+\delta)^2}\right),
\end{align}
the derivative of the covariance matrix with respect to $s$ is written as
\begin{align}
\frac{\partial V_+(s)}{\partial s}=\left[-(v_2-v_1)\sqrt{-\frac{\partial^2 \mu_+}{\partial s^2}}\sigma_x +\eta N_s\delta'(s)|0\rangle\langle 0|\right] \otimes \mathbb{1},
\end{align}
where $\sigma_x$ is the Pauli-$x$ matrix.
One can readily find the solution of Eq.~\eqref{G_eq} for $G$ which is given by
\begin{align}
G=
\begin{pmatrix}
g_{11} & g_{12} \\
g_{21} & g_{22}
\end{pmatrix}
\otimes \mathbb{1}
\end{align}
with
\begin{align}
g_{11}&=\frac{-2\eta N_s\delta'(s)}{4v_1^2-1}, ~~ g_{12}=g_{21}=\frac{-2(v_2-v_1)}{4v_1v_2-1}\sqrt{-\frac{\partial^2\mu_+}{\partial s^2}} \\ 
g_{22}&=0 \text{ if } v_2>1/2, ~~~ g_{22} \text{ is arbitrary if } v_2=1/2.
\end{align}
Thus,
\begin{align}
H_+(s)=2\left[\frac{2\eta^2 N_s^2\delta'(s)^2}{4v_1^2-1}+\frac{4(v_2-v_1)^2}{4v_1v_2-1}\left(-\frac{\partial^2\mu_+}{\partial s^2}\right)\right]
\end{align}
After some simplification of the expression, we obtain the quantum Fisher information from the symmetric mode, which is given by
\begin{align}
H_+(s)=\frac{\eta^2 N_s^2\delta'(s)^2}{(\eta N_s(1+\delta)+N_n+1)(\eta N_s(1+\delta)+N_n)}-\frac{2\eta^2N_s^2[(1+\delta)(\delta''(0)-\delta''(s))+\delta'(s)^2]}{(2N_n+1)(2\eta N_s(1+\delta)+2N_n+1)-1}.
\end{align}
Similarly, one can easily find that
\begin{align}
H_-(s)&=\frac{\eta^2 N_s^2\delta'(s)^2}{(\eta N_s(1-\delta)+N_n+1)(\eta N_s(1-\delta)+N_n)}-\frac{2\eta^2 N_s^2[(1-\delta)(\delta''(0)+\delta''(s))+\delta'(s)^2]}{(2N_n+1)(2\eta N_s(1-\delta)+2N_n+1)-1}.
\end{align}

Let us find the optimal measurement that gives the classical Fisher information equal to quantum Fisher information.
The optimal measurement can be found by diagonalizing the matrix $G$ \cite{oh2019optimal}.
Let us first consider the symmetric mode. The matrix $G_+$ can be diagonalized as
\begin{align}
G_+=
\begin{pmatrix}
g_{11} & g_{12} \\
g_{21} & g_{22}
\end{pmatrix}
\otimes \mathbb{1}_2
=O_+^\text{T}
\begin{pmatrix}
g_1 & 0 \\
0 & g_2
\end{pmatrix}
O_+ \otimes \mathbb{1}_2, ~~~~ \text{where}~~~
O_+=
\begin{pmatrix}
\cos \theta & \sin \theta \\
-\sin \theta & \cos \theta
\end{pmatrix}.
\end{align}
Thus, $G_+$ can be decoupled into two-mode by a beam splitter corresponding to the symplectic matrix $O_+\otimes \mathbb{1}_2$.
To be more specific, the beam splitter angle $\theta$ with the transmittance and reflectance  being $\cos\theta$ and $\sin\theta$ is given by $\theta=1/2\tan^{-1}(2g_{12}/g_{11})$.
Similarly, $G_-$ for anti-symmetric modes can also be decoupled by a beam splitter represented by $O_-\otimes \mathbb{1}_2$, which can be obtained in the same way.

Note that the symmetric logarithmic derivative operator for Gaussian states can be written as \cite{oh2019optimal}
\begin{align}
\hat{L}\propto\hat{Q}^\text{T}G \hat{Q}
\end{align}
with $\hat{Q}=(\hat{x}_1,\hat{p}_1,\hat{x}_2,\hat{p}_2,\hat{x}_3,\hat{p}_3,\hat{x}_4,\hat{p}_4)$. Here, each quadrature operator corresponds to the mode $\hat{a}_+, \hat{b}_+, \hat{a}_-$, and $\hat{b}_-$.
In this case,
\begin{align}
\hat{L}\propto\hat{Q}^\text{T}G \hat{Q}=(O\hat{Q})^\text{T}(\text{diag}(g_1,g_2,g_3,g_4)\otimes \mathbb{1}_2) (O\hat{Q})\propto g_1\hat{n}'_1+g_2\hat{n}'_2+g_3\hat{n}'_3+g_4\hat{n}'_4,
\end{align}
where $\hat{Q}'=O\hat{Q}$, $\hat{n}_i=(\hat{x}_i^2+\hat{p}_i^2-1)/2$, and
\begin{align}
O=(O_+\otimes \mathbb{1}_2)\oplus(O_-\otimes \mathbb{1}_2).
\end{align}
Thus, the photon-number resolving detection after the beam splitters for each two-mode is optimal.

\section{Lower bound of classical Fisher information of fin-SPADE}\label{appendix:CFI}
We calculate the lower bound of classical Fisher information of fin-SPADE method with thermal noise, following the procedure employed in Ref. \cite{nair2016far}.
The difference from Ref. \cite{nair2016far} is the presence of thermal noise.
Let us recall that the lower bound of classical Fisher information for an unknown parameter $\theta$ is given by $F(\theta)\geq \dot{\vec{\mu}}^\text{T}C^{-1}\dot{\vec{\mu}}$,
where $\vec{\mu}$ is the mean vector of the measurement outcome, and $C$ is the covariance matrix of the outcome.
Thus, in the section, we find the mean and the covariance matrix of the measurement outcome from fin-SPADE.

We assume the Gaussian point spread function,
\begin{align}
\psi(x)=\frac{1}{(2\pi\sigma^2)^{1/4}}\exp\left[-\frac{x^2}{4\sigma^2}\right].
\end{align}
Let $h_q$ be a Hermite-Gaussian spatial mode ($q$ is a non-negative integer),
\begin{align}
h_q(x)=\left(\frac{1}{2\pi\sigma^2}\right)^{1/4}\frac{1}{\sqrt{2^q q!}}H_q\left(\frac{x}{\sqrt{2}\sigma}\right)\exp\left(-\frac{x^2}{4\sigma^2}\right).
\end{align}
The quantum state of light in thermal states on the image plane can be written as
\begin{align}
\hat{\rho}=\int d^2 A_1 d^2A_2 p_{N_s}(A)|\psi_{A,s}\rangle\langle \psi_{A,s}|
\end{align}
where
\begin{align}
p_{N_s}(A)=\left(\frac{1}{\pi \eta N_s}\right)^2\exp\left(-\frac{|A_1|^2+|A_2|^2}{\eta N_s}\right)
\end{align}
is the probability density of the source field amplitudes $A=(A_1,A_2)$, and the conditional state $|\psi_{A,s}\rangle$ represents a coherent state with an amplitude
\begin{align}
\psi_{A}(x)=A_1\psi(x-s/2)+A_2\psi(x+s/2).
\end{align}
When thermal noise occurs, the quantum state conditioned on $A$ is changed to
\begin{align}
\psi_{A,\xi}(x)=A_1\psi(x-s/2)+A_2\psi(x+s/2)+\xi(x),
\end{align}
where $\xi(x)$ is a random variable satisfying $\langle\xi(x)\rangle=0$, and $\langle\xi^*(x_1)\xi(x_2)\rangle=N_n\delta(x_1-x_2)$ which describes a random Gaussian displacement noise.
Conditioned on $A$, the amplitude in the $q$-mode can be written as
\begin{align}
B_{q|A,\xi}&=\int_{-\infty}^\infty dx h^*_q(x)\psi_{A,\xi}(x)=\int_{-\infty}^\infty dx h^*_q(x)[A_1\psi(x-s/2)+A_2\psi(x+s/2)+\xi(x)] \\
&=R_q\exp(-Q/2)\frac{Q^{q/2}}{\sqrt{q!}}+\int_{-\infty}^\infty dxh_q^*(x)\xi(x).
\end{align}
where
\begin{align}
\int_{-\infty}^{\infty} dx h^*_q(x)\psi(x+s/2)=(-1)^q\int_{-\infty}^{\infty} dx h^*_q(x)\psi(x-s/2)=(-1)^q\exp(-Q/2)\frac{Q^{q/2}}{\sqrt{q!}}
\end{align}
with $Q=s^2/16\sigma^2$ and $R_q=A_1+A_2$ when $q$ is even, $R_q=A_1-A_2$ otherwise.
Thus, the photocounts $N_{q|A,\xi}$ in each mode are the independent Poisson random variable with the mean
\begin{align}
\mu_{q|A,\xi}=|R_q|^2f_q+\sqrt{f_q}\left(R_q \int_{-\infty}^{\infty} dx h_q^*(x)\xi^*(x)+R_q^* \int_{-\infty}^{\infty} dx h_q(x)\xi(x)\right)+\int_{-\infty}^{\infty} dx_1dx_2h^*(x_1)h(x_2)\xi^*(x_1)\xi(x_2),
\end{align}
where $f_q=\exp(-Q)\frac{Q^{q}}{q!}$,
and the unconditional photocurrent on each mode is written as
\begin{align}
\mu_q=\langle |B_{q|A,\xi}|^2 \rangle_{A,\xi}=2\eta N_s f_q+N_n.
\end{align}
Thus, the derivative of the mean photocurrent is given by
\begin{align}
\frac{\partial \mu_q}{\partial s}=\frac{\eta N_s s}{4\sigma^2}(f_{q-1}-f_q),
\end{align}
with $f_{-1}\equiv 0$.
For the second moments, for $q=q'$, we obtain
\begin{align}
\mathbb{E}[N_q^2]&=\langle \mathbb{E}[N_{q|A,\xi}^2]\rangle_{A,\xi}=\langle \mu_{q|A,\xi}^2+\mu_{q|A,\xi}\rangle_{A,\xi}\\
&=\langle|R_q|^4f_q^2+4|R_q|^2f_q \int_{-\infty}^{\infty} dx_1dx_2h_q^*(x_1)h_q(x_2)\xi^*(x_1)\xi(x_2)+\left(\int_{-\infty}^{\infty} dx_1dx_2h_q^*(x_1)h_q(x_2)\xi^*(x_1)\xi(x_2)\right)^2\rangle_{A,\xi}+\mu_q\\
&=8\eta^2 N_s^2f_q^2 +8\eta N_sN_n f_q+2N_n^2+2\eta N_s f_q+N_n.
\end{align}
When $q\neq q'$ and $q-q'$ is even, we get
\begin{align}
\mathbb{E}[N_q N_q']&=\langle \mathbb{E}[N_{q|A,\xi}N_{q'|A,\xi}]\rangle_{A,\xi}=\langle |B_{q|A,\xi}|^2|B_{q'|A,\xi}|^2\rangle_{A,\xi}\\
&=\langle |R_q|^4f_qf_q'+|R_q|^2f_q \int_{-\infty}^{\infty} dx_1dx_2 h_{q'}^*(x_1) h_{q'}(x_2)\xi^*(x_1)\xi(x_2)+|R_{q'}|^2f_{q'} \int_{-\infty}^{\infty} dx_1dx_2 h_{q}^*(x_1) h_{q}(x_2)\xi^*(x_1)\xi(x_2)\\ 
&+\int_{-\infty}^{\infty} dx_1dx_2dx_3dx_4h_q^*(x_1)h_q(x_2) h_{q'}^*(x_3) h_{q'}(x_4)\xi^*(x_1)\xi(x_2)\xi^*(x_3)\xi(x_4)\rangle_{A,\xi}\\
&=8\eta^2 N_s^2f_qf_{q'} +2\eta N_sN_n (f_q+f_{q'})+N_n^2.
\end{align}
Finally, when $q\neq q'$ and $q-q'$ is odd, we obtain
\begin{align}
\mathbb{E}[N_q N_q']&=\langle \mathbb{E}[N_{q|A,\xi}N_{q'|A,\xi}]\rangle_{A,\xi}=\langle \mu_{q|A,\xi}\mu_{q'|A,\xi}\rangle_{A,\xi}\\
&=\langle |R_q|^2|R_{q'}|^2f_qf_q'+|R_q|^2f_q \int_{-\infty}^{\infty} dx_1dx_2 h_{q'}^*(x_1) h_{q'}(x_2)\xi^*(x_1)\xi(x_2)+|R_{q'}|^2f_{q'} \int_{-\infty}^{\infty} dx_1dx_2 h_{q}^*(x_1) h_{q}(x_2)\xi^*(x_1)\xi(x_2)\\ 
&+\int_{-\infty}^{\infty} dx_1dx_2dx_3dx_4h_q^*(x_1)h_q(x_2)h_{q'}^*(x_3)h_{q'}(x_4)\xi^*(x_1)\xi(x_2)\xi^*(x_3)\xi(x_4)\rangle_{A,\xi}\\
&=4\eta^2 N_s^2f_qf_{q'} +2\eta N_sN_n (f_q+f_{q'})+N_n^2.
\end{align}

Thus, the covariance matrix is written as
\begin{align}
C_{qq'}=
\begin{cases} 
      4\eta^2 N_s^2f_q^2+4\eta N_sN_nf_q+2\eta N_sf_q+N_n^2+N_n & q=q' \\
      4\eta^2 N_s^2f_qf_{q'} & q\neq q' \text{ and } q-q' \text{ is even} \\
      0 & q\neq q' \text{ and } q-q' \text{ is odd}
   \end{cases}
\end{align}

The covariance matrix and the derivative of the first moment give the lower bound of classical Fisher information as stated in the main text.

\end{widetext}

\bibliography{reference.bib}

\begin{thebibliography}{48}
\bibitem [{fin()}]{finite}%
  \BibitemOpen
  \href@noop {} {}\bibinfo {note} {We assume that the series always converge to
  focus on physically relevant situations.}\BibitemShut {Stop}%
\end{thebibliography}

\end{document}